\documentclass[preprint, 12pt, proof, authoryear]{elsarticle}
\usepackage[left=1in, right=1in, top = 1.5in, bottom = 1.5in]{geometry}
\usepackage[round]{natbib}
\usepackage[utf8]{inputenc}
\usepackage[all]{xy}
\usepackage{amsthm,amssymb,centernot, amsfonts}
\usepackage{amsmath}
\usepackage{csquotes}
\usepackage{subfig}
\usepackage{enumerate}
\usepackage{graphicx}
\usepackage{enumerate}
\usepackage{color}
\usepackage{mwe}
\usepackage{appendix}
\usepackage{centernot}
\usepackage{url,hyperref,xspace}
\usepackage[ruled,vlined]{algorithm2e}
\usepackage{soul}
\usepackage{xcolor}
\usepackage{mathtools}
\usepackage{footnote}
\usepackage{tablefootnote}
\usepackage{caption}
\usepackage{adjustbox}
\usepackage{booktabs}
\usepackage[normalem]{ulem}
\usepackage{cancel}

\makesavenoteenv{tabular}
\usepackage[flushleft]{threeparttable}
\def\@centernot#1#2{%
  \mathrel{%
    \rlap{%
      \settowidth\dimen@{$\m@th#1{#2}$}%
      \kern.5\dimen@
      \settowidth\dimen@{$\m@th#1=$}%
      \kern-.5\dimen@
      $\m@th#1\not$%
    }%
    {#2}%
  }%
}
\makeatother

\newcommand{\independent}{\perp\mkern-9.5mu\perp}

\usepackage{tikz}
\usetikzlibrary{arrows.meta,shapes}
\usetikzlibrary{arrows,shapes.arrows,shapes.geometric,shapes.multipart,
decorations.pathmorphing,positioning,shapes.swigs,}

\usepackage{setspace}
\newtheorem{proposition}{Proposition}

\newtheorem{definition}{Definition}

\MakeOuterQuote{"}\EnableQuotes
\DeclareUnicodeCharacter{00A0}{ }

\usepackage[final]{changes}

\makeatletter
\def\ps@pprintTitle{%
 \let\@oddhead\@empty
 \let\@evenhead\@empty
 \def\@oddfoot{}%
 \let\@evenfoot\@oddfoot}
\makeatother

\begin{document}

\begin{frontmatter}
\title{Perspectives on \replaced{`harm'}{harm} in personalized medicine}

\author{Aaron L. Sarvet  \corref{cor1}}
\author{Mats J. Stensrud}

\address{Department of Mathematics, École Polytechnique Fédérale de Lausanne, Switzerland}
\cortext[cor1]{\textbf{Contact information for corresponding author:}\\
Aaron L. Sarvet, Department of Mathematics, École Polytechnique Fédérale de Lausanne, Switzerland. \url{aaron.sarvet@epfl.ch}}

\begin{abstract}
Avoiding harm is an uncontroversial aim of personalized medicine and other epidemiologic initiatives. However, the precise mathematical translation of ``harm'' is disputable. Here we use a formal causal language to study common, but distinct, definitions of ``harm''. We clarify that commitment to a definition of harm has important practical and philosophical implications for decision making. We relate our practical and philosophical considerations to ideas from medical ethics and legal practice.   
\end{abstract}
\end{frontmatter}

\newpage

\section{Introduction} \label{sec:intro}

Personalized medicine aims to leverage observed data to personalize treatment assignment in the future. When treatments are personalized, their assignment depends on patient characteristics.  Individualizing treatments is justifiable if it results in better (expected) utility in the population. In some cases, personalized treatments are (biologically) engineered to be compatible with known patient features, for example with immunologically-targeted cancer therapies. In other cases, several treatments already exist, and the optimal treatment for each patient can be learned from data. Finding the optimal treatment requires causal inference: we need to identify and subsequently estimate \added{the} patient's expected utilities under each of the treatment options that are available \citep{robins1986new,richardson2013single,murphy2003optimal, robins2004optimal,tsiatis2019dynamic,kosorok2021introduction}. 

Specifying utility \textit{functions} -- the mapping of a patient's outcomes and other features to some number -- is a fundamental step in posing this task, because it forms the basis for determining preferences among treatment options. However, this step is often avoided by epidemiologists and statisticians, perhaps because of its philosophical and subjective reputation. Yet, decision makers and methodologists will usually agree,  at least superficially, with the classical Hippocratic maxim: ``\textit{First, do no harm}''. 

Consideration of ``harm'' is the focus of this article. We anchor results in an emerging literature offering an  operationalization of ``harm'' that promises to possibly `provide the key for next-generation personalized decision making.' \citep{mueller2022personalized}.  Their methods propose strategies for supplementing experimental data with confounded non-experimental data that are often disregarded in classical approaches \citep{tian/pearl:probcaus}. We demonstrate that their proposals take for granted a particular operationalization of ``harm'', and we contrast it with notions of harm defined by parameters identified by experimental data. To bring clarity, we articulate two alternative approaches, which we refer to as \textit{counterfactual} and \textit{interventionist}, respectively. These two approaches are distinguished by utility functions of fundamentally different natures. Furthermore, they each correspond to different notions of ``harm'' and to different interpretations of the Hippocratic principle.

\section{Preliminaries} \label{sec: prelim}  
Suppose we observe a random sample of patients from a large population. We let $Y$ be the outcome, $A$ be a binary treatment, \added{and }$L$ be baseline covariates\deleted{, and $R$, a binary indicator of participation within a controlled experiment on $A$ within levels of $L$}.  

Throughout, we will use potential outcomes \textit{notation} to denote the outcome that would occur under an intervention that fixes treatment to a specific value \citep{dbr:jep} (see Web Appendix 1 for additional remarks). For example, $Y^{a=1}$ denotes a patient's potential outcome that would occur had $A$ been fixed, possibly counter to fact, to $a=1$.

We \replaced{consider $\{L, A, Y, Y^{a=1}, Y^{a=0}\}$ to be the full data, of which only $\{L, A, Y\}$ are observed.}{let $O_F \equiv \{R, L, A, Y, Y^{a=1}, Y^{a=0}\}$ denote the full data, and $O \equiv \{R, L, A, Y\}$denote the observed data.}  \deleted{We let $P$ denote an arbitrary joint distribution for the full data, and we let $P\equiv P(P)$ denote the subdistribution of $P$ for the observed data. Furthermore, we let $P$ and $P$ denote the true such distributions at the time of the study, and further let $P$ and $P$ be the analogous such distributions at the time of a future decision. We distinguish between $P$ and $P$ to formalize assumptions used to relate observed data inferences to future decisions.}When discussing the \textit{counterfactual} approach we use $S$ to denote the joint principal \replaced{stratum}{strata} variable $\{Y^{a=1},  Y^{a=0}\}$ \citep{robins1986new, frangakis2002principal}, where the events $(S=1, S=2, S=3, S=4)$ correspond to the events $(\{Y^{a=1}=1,  Y^{a=0}=0\}, \{Y^{a=1}=0,  Y^{a=0}=1\}, \{Y^{a=1}=1,  Y^{a=0}=1\}, \{Y^{a=1}=0,  Y^{a=0}=0\})$. 

\deleted{We define a model $\mathcal{M}$ based on assumptions listed in Table \ref{tab: M1}. We consider $\mathcal{M}$ to be minimal because it comprises only the basic model motifs that are used to: 1) link observed and causal parameters ($\textbf{A1}$); 2) distinguish experimental and non-experimental data ($\textbf{A2}$); 3) justify the fusion of experimental and non-experimental data ($\textbf{A3}$); and 4) justify the use of observed data for future decision settings ($\textbf{A4}$). [TABLE 1 REMOVED]}

\subsection{Defining utility functions}

Utilities are crucial to theories about decisions \citep{von2007theory}. The simplest utility function is one that is equal to a patient's outcome under a treatment $a$, $Y^a$. A more elaborate utility function might depend further on the treatment actually received. We will consider one such utility $U^{a}_1$ under $a$,

$$U_1^a = \mu(Y^a, a).$$  

This utility is defined solely by variables that are directly observable under treatment $a$. We will refer to these utilities as \textit{interventionist utilities}.

In principle, however, utilities may be functions of any subset of well-defined patient variables, for example, the potential outcome corresponding to a treatment that was \textit{not} assigned to that patient. In contrast to an \textit{interventionist} utility, this utility is defined in part by \textit{counterfactual} variables that would not be directly observable under treatment $a$. We will thus refer to these utilities as \textit{counterfactual} utilities. A special case is a $U_2^a$ that depends simultaneously on the interventionist utilities that would arise under either treatment level:

\begin{align}
    U_2^a = & \gamma^*(\mu(Y^a, a), \mu(Y^{1-a}, 1-a), a) \label{eq: gammastar}\\
          =  & \gamma(S, a). \label{eq: gamma}
\end{align}

$U_2^a$ may alternatively be understood as a function of a patient's principal strata $S=\{Y^{a=1}, Y^{a=0}\}$ and their assigned treatment $a$. A patient or investigator specifying a counterfactual utility may either specify $\mu$ and $\gamma^*$, as in \eqref{eq: gammastar}, or alternatively may specify $\gamma$ directly, as in \eqref{eq: gamma}. 

We refer to the \textit{counterfactual} and  \textit{interventionist} approaches as those that use \textit{counterfactual} and  \textit{interventionist} utility functions, respectively.






\section{Twin definitions of harm in personalized medicine} \label{sec: harm} 


Recital of the Hippocratic Oath ``first, do no harm'' is a common ritual for clinical trainees \citep{orr1997use}. However, there is broad consensus that the Oath is not a sufficiently rich formulation of a modern medical ethics, nor does its meaning unambiguously speak for itself \citep{beauchamp1994principles}. Nevertheless, this has not prevented \added{some} contemporary authors in causal inference from claiming a singular mathematical translation: for example \citet{ben2022policy} identify the ``first, do no harm'' principle with a particular specification of a \textit{counterfactual} utility function, $\gamma$, such that 



\begin{align}
    \gamma(s=1, a=0) - \gamma(s=1, a=1) > \gamma(s=2, a=1) - \gamma(s=2, a=0). \label{eq: benharm}
\end{align}

When $Y$ indicates patient death, then $S=1$ indicates a patient that would only die under treatment (the ``harmed''), and $S=2$ indicates a patient who would only survive under treatment (the ``saved'').  Thus, \citet{ben2022policy} consider a $\gamma$ such that the utility \textit{gain} upon \textit{withholding} treatment from a patient who would be ``harmed'', $\gamma(s=1, a=0) - \gamma(s=1, a=1)$, is greater than that \textit{gain} upon \textit{giving} treatment to a patient who would be ``saved'', $\gamma(s=2, a=1) - \gamma(s=2, a=0)$. 

These \replaced{appear to be}{may seem like} symmetric actions and thus should be assigned equal \textit{gain}: both decisions presumably save a life, by selecting the only decision resulting in a favorable outcome. However, according to \replaced{the translation of harm adopted by}{this translation} \citep{ben2022policy}, treatment $a=1$ is privileged with the capacity to ``do''. Thus, the first \textit{gain} apparently avoids an \textit{error of comission}, whereas the second \textit{gain} avoids an \textit{\textit{error of omission}}. The asymmetry apparently derives from the refrain to ``\underline{\textit{first}}, do no harm'', as in ``\underline{\textit{first}}'' make no \textit{errors of comission}. \citet[pg. ~8]{ben2022policy} assert that this particular type of asymmetry is ``encoded by the principle to 'do no harm' in the Hippocratic oath''. Their translation motivates the following definition of \textit{counterfactual harm}:

\begin{definition}[Counterfactual harm] \label{def: CFharm}
Suppose a patient is assigned treatment $a=1$ (relative to no treatment, $a=0$) and that their interventionist utility under treatment is less than their utility under no treatment, $\mu(Y^{a=1}, a=1) < \mu(Y^{a=0}, a=0)$. This is \textit{counterfactual harm}. 
\end{definition}

In other words, \textit{counterfactual harm} occurs when a patient takes treatment $a=1$ but their \replaced{personal}{own} \textit{interventionist utility} is not maximized by this treatment value.  We call this \textit{counterfactual} because its definition depends on at least one potential outcome for a patient that is in principle impossible to observe. We associate this notion of \textit{counterfactual harm} with the proposals of \citet{ben2022policy} and \citet{li/pearl:ijcai19, mueller2022personalized}, because their proposals also require consideration of a patient's principal stratum.




While Definition \ref{def: CFharm} is an ostensibly natural interpretation of ``harm'' in the Hippocratic sense, examination of modern medical ethics suggests that an alternative formulation may operate in practice. 
\citet{beauchamp1994principles} noted that the principle of \textit{Nonmaleficence} is treated identically to the principle of ``first, do no harm''. However, they distinguish between cases where negative outcomes are presumably caused almost surely, and cases where merely the \textit{risk} of harm is caused by an agent. For example, \citet[pg. ~154]{beauchamp1994principles} state that ``Obligations of nonmaleficence include not only obligations not to inflict harms, but also obligations not to impose risks of harm.'' In such cases where certainty is forgone, \textit{maleficence} is identified with a deviation from a ``standard of due care''. But in modern evidence-based medicine, several standards of care are often available, i.e., whenever a second treatment $a=1$ is accepted as non-inferior to a first treatment $a=0$. These determinations are often based on clinical trials, in which $\mathbb{E}[Y^{a=1}]$ and $\mathbb{E}[Y^{a=0}]$ are identified. 


Consider then the following example: suppose $U^a_1=(1-Y^{a})$ defined a particular \textit{interventionist} utility and thus \replaced{$\{Y^{a=1} > Y^{a=0}\} \equiv \{S=1\}$}{$I(Y^{a=1} > Y^{a=0}) \equiv I(S=1)$} constitutes \textit{counterfactual harm}. It is possible that both $P(S=1)>0$ \textit{and} $\mathbb{E}[Y^{a=1}]=\mathbb{E}[Y^{a=0}]$; some patients are \textit{counterfactually harmed} by treatment in the sense of Definition \ref{def: CFharm}, yet a physician could still be considered to have followed ``due standard of care'' when assigning $a=1$. This has an immediate, important implication: if upholding a ``due standard of care'' is equivalent to upholding the principle of \textit{nonmaleficence}, and if \textit{nonmaleficence} is equivalent to the principle of ``first, do no harm'', then we cannot simply equate ``harm'' in the Hippocratic sense with \textit{counterfactual harm} in Definition \ref{def: CFharm}. This argument motivates an alternative definition more concordant with \textit{nonmaleficence} in modern bioethics when outcomes are uncertain, which is nearly always the case. To state this definition of harm, we first define the patient features that a provider must consider in order to uphold ``due standard of care''.

\begin{definition}[Relevant features] \label{def: relfeat}
Let $W$ denote the features taking well-defined values prior to treatment. Let $X$ denote the subset of $W$ that is \textit{measured} prior to treatment and suppose a patient has a realization $X=x$. Let $X'$ be the largest subset of $X$ such that, for $x'\subset x$,  $$\text{sign}\Big\{\mathbb{E}[Y^{a=1} \mid X'=x'] - \mathbb{E}[Y^{a=0} \mid X'=x']\Big\}$$  is taken to be fixed and known by the investigator for the purpose of the decision task. Then $X'$ are the \textit{relevant features} for that patient.  
\end{definition}

Definition \ref{def: relfeat} identifies a patient's features that can be practically used to make a personalized comparison of expected outcomes under either treatment. For practical purposes, the sign of the ATE for patients with relevant features $x'$ might be taken as fixed and known upon consideration of subject matter knowledge, the results of sufficiently large clinical trials, or other empirical evidence (e.g., consistent estimates of conditional causal effects from the data on hand.) Now we can \replaced{give}{provide} a second definition of ``harm'' that we refer to as \textit{interventionist harm}.

\begin{definition}[Interventionist harm] \label{def: Intharm}
Suppose  $L$ represents the \textit{relevant features} for a patient. Suppose a patient is assigned treatment $a=1$ (relative to no treatment, $a=0$) and that the expected potential utility under treatment, given $L=l$, is less than that expected utility under no treatment, $$\mathbb{E}[\mu(Y^{a=1}, a=1) \mid L=l] < \mathbb{E}[\mu(Y^{a=0}, a=0) \mid L=l].$$ This is \textit{interventionist harm}. 
\end{definition}

Thus, \textit{interventionist harm} occurs when a patient takes treatment $a=1$ but their expected \textit{interventionist utility} conditional on their \textit{relevant features} is not maximized.  This is an \textit{interventionist} notion because the parameters needed to determine its realization are observable in experiments. Furthermore, it motivates utilities based on factual responses  observable post-hoc for a given patient.

\section{Decision theory}  \label{sec: decision} 

It is classically assumed that rational agents in a future decision setting will select the treatment $a$ maximizing their expected utility \citep{von2007theory}, $\mathbb{E}[U^a]$\deleted{, which is a parameter of $P$}. A patient would leverage both their knowledge of their own features\deleted{, i.e.\ realizations of $O$,} at the decision point, and knowledge of \replaced{the distribution of the full data}{$P$}; that is, they will base their decisions on the parameters of $P$ corresponding to their \textit{relevant} features. For example, when a patient observes the value of their baseline relevant features $L=l$, then they would choose the treatment $a$ equal to $$\text{argmax}_{a}\mathbb{E}[U^a \mid L=l].$$

Suppose a patient with relevant features $X'=x'$ subscribes to a \textit{counterfactual} notion of harm (Definition \ref{def: CFharm}). Thus, they specify a counterfactual utility function $\gamma$ via specification of an interventionist utility function $\mu$ and its modifying function $\gamma^*$. Let $U^a_2$ define this utility. If this patient is a rational agent, they will choose the treatment $a$ maximizing $$\mathbb{E}[U_2^a \mid X'=x'] = \sum_{s} \gamma(s, a)P(S=s \mid X'=x').$$

Alternatively, suppose this patient rather subscribes to an \textit{interventionist} notion of harm (Definition \ref{def: Intharm}). Thus, they merely specify an \textit{interventionist} utility function $\mu$. Let $U^a_1$ denote this \textit{interventionist} utility under $\mu$. If this patient is a rational agent, they will choose the treatment $a$ maximizing $$\mathbb{E}[U_1^a \mid X'=x'] = \sum_{y} \mu(y, a)P(Y^{a}=y \mid X'=x').$$

\section{Practical considerations}\label{sec: pract} 

\subsection{Utility maximization with experimental data}

Consider the practical task of selecting a treatment using parameters of the observed data\deleted{$P$, under model $\mathcal{M}$ defined in Table \ref{tab: M1}}. For a patient maximizing \textit{interventionist} utility, this is reducible to identifying $\mathbb{E}[Y^a \mid L=l]$ for each $a$ in terms of parameters of those observed data\deleted{ at the time of data collection, $P$}. \replaced{When data arise from an ideal experiment, we}{We} identify $\mathbb{E}[Y^a \mid L=l]$ by \deleted{from parameters of $P$ }

\begin{align}
    \mathbb{E}[Y^a \mid L=l] = & \mathbb{E}[Y \mid L=l, A=a]. 
\end{align}

This identity follows from assumptions \replaced{enforced by design of the experiment}{of model $\mathcal{M}$} (see Web Appendix 2 \added{for explicit arguments}).  With these parameters and a specified \textit{interventionist} utility function $\mu$, then $L$ would constitute each patient's relevant covariate. A patient with $L=l$ can readily choose the treatment $a$ maximizing expected utility via $$\text{argmax}_a\mathbb{E}[Y \mid L=l, A=a].$$ 

For a patient maximizing \textit{counterfactual} utility, this task is reducible to identifying $P(S=s \mid L=l)$ for each $s$\deleted{ in terms of $P$}, but is more challenging \added{than maximizing \textit{interventionist} utility}; identification is generally not possible under \replaced{only those assumptions enforced by}{model $\mathcal{M}$, which we note describes data from} an ideal randomized trial.

Alternatively, the values of $P(S=s \mid L=l)$ may be \textit{partially} determined under \replaced{such assumptions}{$\mathcal{M}$} by leveraging constraints imposed by identified quantities. Specifically these constraints can be used to construct logical upper and lower limits, i.e., \textit{bounds}, on the principal stratum probabilities, which fully leverage the\added{se} assumptions\deleted{ in $\mathcal{M}$}. When upper and lower limits coincide, probabilities are identified. Otherwise, the true value of the probabilities within the bounds will depend ultimately on a parameter whose value is completely uninformed by \replaced{the observed data in this setting}{$P$ under the model}.  

We review two cases when bounds are sufficient for identifying the treatment $a$ maximizing $\mathbb{E}[U_2^a \mid L=l]$. A first case is when the upper and lower limits coincide, so that the joint distribution of $\{Y^{a=1},Y^{a=0}\}$ is fully determined by the experimental data, which occurs \deleted{under $\mathcal{M}$} only in the case that, \deleted{for some $a,y =0,1$, $P(Y^a=y \mid L=l) =0$.  That is to say, }for at least one treatment \added{level $a$}, the resulting outcome is predicted with certainty.  \deleted{In this case, both joint events $(Y^a = y, Y^{1-a}=0)$ and $(Y^a = y, Y^{1-a}=1)$ have probability $0$ among those with $L=l$.} Thus, the conditional distribution of $S$ is point identified, and $\mathbb{E}[U_2^a \mid L=l]$ is evaluated directly. This is a condition that is verifiable with \deleted{knowledge of $P$ whenever $P$ describes}experimental data, i.e. when $P(Y^a=y \mid L=l)$ is identified\deleted{ by $P$, which is the case under $\mathcal{M}$}. 

Second, certain specifications of $\gamma$ permit maximization even when only bounds for $P(S=s \mid L=l)$ are available. If we take $\delta(s) = \gamma(s, a=1) - \gamma(s, a=0)$ to be the counterfactual utility \textit{gain} for choosing treatment $a=1$ among patients with principal stratum $s$, we can re-express $\mathbb{E}[U_2^{a=1} \mid L=l] - \mathbb{E}[U_2^{a=0} \mid L=l]$ as 

\begin{align}
    \sum_s \delta(s)P(S=s \mid L=l). \label{eq: diff}
\end{align}

Particularly, when $\gamma$ is such that  

\begin{align}
    \delta(1) + \delta(2) = \delta(3) + \delta(4), \label{eq: gaineq}
\end{align}

then $\mathbb{E}[U_2^{a=1} \mid L=l] - \mathbb{E}[U_2^{a=0} \mid L=l]$ can be re-expressed as

\begin{align} 
      & \Bigg\{\Big(\delta(1)-\delta(4)\Big)P(Y^{a=1}=1 \mid L=l)\Bigg\} +  \Bigg\{\Big(\delta(3)-\delta(1)\Big)P(Y^{a=0}=1 \mid L=l)\Bigg\}  + \delta(4) \label{eq: gaineq2}. 
\end{align}

Expression \eqref{eq: gaineq2} is only a function of the investigator-specified $\delta$, and terms like $P(Y^{a}=y \mid L=l)$, which are identified \replaced{in an experiment}{under $\mathcal{M}$}. The constraint on $\gamma$ in expression \eqref{eq: gaineq} has been called ``gain equality'' by some authors \citep{li/pearl:ijcai19, li2022unit3}. \citet{ben2022policy} consider a special class of $\gamma$,
called ``symmetric'' utility functions, which also satisfy ``gain equality''. Such functions weigh equally errors of ``comission'' and errors of  ``omission'' and thus would not satisfy their  interpretation of the ``first, do no harm principle``. 


In general, however, patients who would like to choose the treatment $a$ maximizing their \textit{counterfactual} utility\replaced{, using only assumptions that could be experimentally enforced,}{ under model $\mathcal{M}$} will have to proceed with only the knowledge that  $P(S=s \mid L=l)$ must be contained within an interval. Building on results in \citet{cui2021individualized},  \citet{ben2022policy} illustrate that, while the treatment value $a^*$ that maximizes $\mathbb{E}[U_2^a \mid L=l]$ is not in general identified, one can identify the value of treatment $a'$ that minimizes the \textit{worst case regret}, relative to $a^*$ \citep{savage1951theory}, and suggest that such a selection procedure be used by a patient. \citet{li/pearl:ijcai19} heuristically suggest an identical procedure, based on the results of two simulated examples. Although regret minimization allows decisions to be made under uncertainty, we note that its use represents a departure from classical decision theory.  

\subsection{On the fusion of experimental \& non-experimental data}

For illustrative purposes, we will suppose that no baseline covariates have been measured, i.e.\ that $L\equiv \emptyset$, but the arguments can straightforwardly be extended to settings with measured baseline covariates.

\subsubsection{Counterfactual approach}

\citet{tian/pearl:probcaus} illustrate how bounds for $P(S=s)$ identified by experimental data can be narrowed by supplementing with confounded non-experimental data. To provide some intuition for this result, consider an arbitrary treatment regime $g$. The average treatment effect (ATE) comparing ``always treat'', $a=1$, and $g$, that is $$\mathbb{E}[Y^{a=1}] - \mathbb{E}[Y^{g}]$$ is a valid lower bound for $P(S=1)$ (See Proposition \ref{prop:s3} in Web Appendix 2). As an example consider the special case where $g$ is ``never treat'', $a=0$. \citet{greenland1986identifiability} classically illustrated that, for a binary $Y \in \{0,1\}$, $$\mathbb{E}[Y^{a=1}] - \mathbb{E}[Y^{a=0}] = P(S=1) - P(S=2).$$ Thus, $P(S=1)$ can be no less than this ATE, and it is only greater than this ATE by an amount equal to $P(S=2)$. \replaced{Among}{Whenever exchangeability holds, it is this special case, among} all the average treatment effects identified by the experimental data, \added{it is this special case} which \textit{maximizes} the lower bound (See Proposition \ref{prop:s4} in Web Appendix 2). For other regimes $g$, which possibly depend on unmeasured features, an improved lower bound could be obtained if not for the fact that no such regime is identified by \replaced{such}{the experimental} data.  Thus, it is this canonical ATE that is typically used to compute sharp nonparametric bounds. However, when \textit{confounded} \textit{non}-experimental data \replaced{are}{is} additionally used, the expected outcome under one such regime \textit{is} \added{indeed} identified: the \textit{factual} (but unknown) regime that was used to generate that non-experimental data, which is $\mathbb{E}[Y]$. This parameter then may be used as one additional chance to improve the lower bound on $P(S=1)$, as illustrated by \citet{tian/pearl:probcaus}. We provide an articulation of these improved bounds in terms of our notation and \replaced{conventional}{the} assumptions \deleted{of model $\mathcal{M}$ }in Web Appendix 3.

While this result has been known for decades, it is currently emphasized by \citet{li2022probabilities1} and also \citet{mueller2022personalized} as an important practical tool for implementing a \textit{counterfactual} approach.

\subsubsection{Interventionist approach}

The allocation of treatment to patients in practice can be decomposed into two steps: 1) treatment \textit{intention}; and 2) treatment \textit{assignment}. In this first step, the patient, or provider, decides on the desired treatment.  This decision will typically be related to health status and other information that could predict the outcome, so that those who do and do not intend to take treatment are not exchangeable groups.  This is confounding. We introduce a binary ``intention to treat'' (ITT) variable $A^*$ to denote the treatment desired. The \textit{assigned} treatment $A$ is the treatment actually received. In non-experimental (confounded) data, the \textit{intended} and \textit{assigned} treatment may always coincide, because patients and clinicians are free to exercise their intentions. In experimental data, one would in general not assume that $A^*$ will coincide with $A$ because in an experiment, a patient's treatment intentions are subverted by the random treatment assignment mechanism. In both cases the observed data will not in general include measurements of $A^*$. However, we might leverage $A^*$ through \replaced{auxiliary}{the additional} assumptions \added{that would justify the fusion of experimental and non-experimental data, which might hold if experimental and non-experimental patients were recruited randomly from the same source population.}  \deleted{where we let $P^{F*}$ denote the law of the full data including the intention to treat variable.} We define \replaced{such assumptions in Web Appendix 2}{in Table \ref{tab: M2} an elaborated model $\mathcal{M}^*$ encoding these assumptions [TABLE 2 REMOVED]}.

Importantly, when $A^*$ is measured in a future decision setting, and knowledge of the sign of the average treatment effect conditional on $A^*$ is known or identified, then $A^*$ will be a \textit{relevant feature} as in Definition \ref{def: relfeat}. \added{Indeed, these effects are identified when properly fused experimental and non-experimental data are available (see Web Appendix 2). Intuitively, the challenging quantity to identify is the expected outcome under assigned treatment level $a$ among patients whose own treatment intention $a'$ -- unmeasured in the experiment -- disagreed with their assignment, $a\neq a'$.  But we also know that the marginal expected outcome under assigned treatment level $a$ -- identified in the experiment -- is simply an average of this conditional quantity and the corresponding expected outcome among those whose own treatment intention actually agreed with their assignment, $a=a'$, weighted by the probabilities of treatment intentions. Since this second set of quantities are observed directly in the non-experimental data, then the challenging quantity is the sole unknown, and its value can be simply deduced.}

As illustrated in \citet{stensrud2022optimal}, a patient will do as well or better, in terms of expected \textit{interventionist} utility, when the ITT variable $A^*$ is additionally used as a relevant covariate. Furthermore, Proposition \ref{prop:s5} in Web Appendix 2 shows that non-experimental data will improve an \textit{interventionist}'s expected utilities if and only if these non-experimental data improve the lower bound on the probability of ``harm'', that is $P(S=1 \mid L=l)$. Thus, the \textit{counterfactual} and \textit{interventionist} uses of confounded non-experimental data are connected.


\section{Philosophical considerations} \label{sec: philosoph} 

We ground philosophical considerations by assuming a simple interventionist utility $U^a_1=\mu(Y^a, a)=1-Y^a$ and suppose that $Y=1$ indicates patient death. We also consider a \textit{counterfactual} utility $U^a_2$ induced by $\gamma^*$ and $\mu$. \added{We emphasize that these considerations are not exhaustive, but are prompted by claims from contemporary advocates of a counterfactual approach. In contrast to these advocates, we do not here promote either approach.}

\subsection{Excess death under a counterfactual utility}
\citet{ben2022policy} show that the task of selecting a treatment value $a$ under a \textit{counterfactual} approach may be reformulated as the following constrained optimization task:

\begin{align*}
    \text{argmin}_a & \mathbb{E}[Y^a] \text{ subject to the constraints that}\\
                    & P(S=1) < \chi_1 \text{ and}\\
                    & P(S=2) < \chi_2,
\end{align*}

where $\chi_1$ and $\chi_2$ are some numbers in $[0,1]$ implied by $\gamma$ under a particular \replaced{distribution of the full data}{$P$}. In contrast, the \textit{interventionist's} task will simply have the optimization task $\text{argmin}_a \mathbb{E}[Y^a]$, subject to \textit{no} constraints. 

This formulation has the following simple implication: when policy-makers optimize \textit{counterfactual} utilities, then, in general, more people will die. 

Proponents of a \textit{counterfactual} approach may argue: but all deaths are not considered equal; they may argue that the true utility, given by $\gamma$, is one that uses the possibly asymmetric \textit{counterfactual} utility function on the principal strata that is appropriately formulated to reflect notions of \textit{counterfactual harm}. However, a serious problem is that we have no direct evidence that these principal strata exist. Even if one makes a metaphysical commitment to their existence, a patient will never know their true principal stratum, except under extreme circumstances; thus, no patient will ever know their post-hoc utility, and no policy may ever be evaluated, or compared to an alternative policy, using direct observations. 

In other words, the \textit{counterfactual} \replaced{approach}{framework} requires a belief in metaphysical objects whose existence can neither be confirmed nor denied. While patients may be free to hold such beliefs, policy-makers should know the implications: when a \textit{counterfactual} framework is deployed to determine social policies and regulations, it coerces conformity to an unverifiable metaphysics and a corresponding logic that deals in those terms. In contrast, when an \textit{interventionist} framework is deployed in such a setting, no such coercion is made, and \deleted{importantly,} patient and group outcomes are observable and thus transparent. 

\subsection{Medical liability in United States Law}
Proponents of a \textit{counterfactual} approach, see e.g. \citet{tian/pearl:probcaus},  \replaced{have justified}{will often justify} its use by appealing to discourse in modern law and medicine. Indeed, the modern medico-legal system in the United States appears to parlay in terms of counterfactual logic when determining liability for patient outcomes. Legal case precedent has established a ``but for'' doctrine in medical malpractice determination. One interpretation asserts that a provider cannot be considered liable for a negative patient outcome except when it could have been determined at the time of treatment that the negative outcome would not occur ``but for'' the provider's actions. However, although such counterfactual terms are apparently invoked, the exact same legal phenomena might be understood through an \textit{interventionist} lens. If liability requires a determination that $P(S=1 \mid L=l)=1$, merely a high probability of ``but-for causation'' is not sufficient, see \citet{Kline} for a review of relevant case law. But, $P(S=1 \mid L=l)=1$ occurs precisely when $P(Y^{a=0}=0 \mid L=l)=1$ and  $P(Y^{a=1}=1 \mid L=l)=1$. Thus, despite the use of counterfactual language, the actual phenomena produced by the current medico-legal framework could equivalently be explained by an \textit{interventionist} logic: a provider is held liable only when $P(Y^{a=0}=0 \mid L=l)=1$ and  $P(Y^{a=1}=1 \mid L=l)=1$. %

\subsection{Issues with partial identification}

An investigator adopting minimal assumptions will in general only  partially identify expected \textit{counterfactual} utilities (see Section \ref{sec: pract}). This is not resolved through supplementation with non-experimental data, since these data will in general only result in narrower bounds. Thus, the treatment that truly maximizes \textit{counterfactual} utility will still be ambiguous. \citet{ben2022policy} proposes to resolve this ambiguity through a regret-minimization approach \citep{savage1951theory}. However, regret minimization is only one of many well-understood approaches for resolving such ambiguity. Alternatives include choosing the rule that maximizes the minimum possible utility, or placing a distribution over the remaining possible values of an unidentified nuisance parameter (see \citet{stoye2012new} and \citet{cui2021individualized} for reviews). These different approaches will in general converge to different decisions at increasing sample sizes whenever \textit{counterfactual} utilities are not point identified, and so an investigator's choice among them is one additional subjective researcher degree of freedom that does not burden an \textit{interventionist}.

\section{Conclusion} \label{sec: conc}
We have defined and contrasted \textit{counterfactual} and \textit{interventionist} approaches to decision making. Contrary to claims of some authors \citep{li/pearl:ijcai19, mueller2022personalized}, a \textit{counterfactual} approach should not \added{necessarily} portend a revolution in personalized medicine. Specifying utility functions and their domains, as well as a framework for decision-making, remains highly subjective. In tension with proponents of a \textit{counterfactual} approach, we have reviewed several practical and philosophical considerations that seem to problematize its use and challenge some of its core premises, e.g.\ that it somehow naturally corresponds with prevailing medical ethics and legal practice. Perhaps most problematic from a population policy-maker's perspective: when the outcome is death and a \textit{counterfactual} approach is used, in general, more people will die under the identified optimal policy compared to that identified by an \textit{interventionist} approach. A strong critique of the \textit{\textit{counterfactual}} approach calls it ``dangerously misguided'' and warns that ``real harm will ensue if it is applied in practice'' \citep{dawidharm}. We take the following stance:  as causal inference become increasingly embedded in the development of personalized medicine, it is important that researchers, practitioners and other decision makers clearly understand the different approaches to decision making and their practical and philosophical consequences.

\section*{Acknowledgements}
We thank Philip Dawid and Stephen Senn for many insightful discussions and comments. \added{This work has been supported by the Swiss National Science Foundation (Grant Number $200021_207436$: Causal Inference when resources are limited).}

\bibliographystyle{plainnat}
\bibliography{refs}

\begin{thebibliography}{28}
\providecommand{\natexlab}[1]{#1}
\providecommand{\url}[1]{\texttt{#1}}
\expandafter\ifx\csname urlstyle\endcsname\relax
  \providecommand{\doi}[1]{doi: #1}\else
  \providecommand{\doi}{doi: \begingroup \urlstyle{rm}\Url}\fi

\bibitem[Beauchamp and Childress(1994)]{beauchamp1994principles}
Tom~L Beauchamp and James~F Childress.
\newblock \emph{Principles of biomedical ethics}.
\newblock Edicoes Loyola, 1994.

\bibitem[Becker et~al.(2018)Becker, Specter, and Kline]{Kline}
Charles~L. Becker, Shanin Specter, and Thomas~R. Kline.
\newblock The supreme court and medical malpractice law.
\newblock In \emph{The Supreme Court of Pennsylvania: Life and Law in the
  Commonwealth, 1684-2017}, pages 241--258. Penn State Press, 2018.

\bibitem[Ben-Michael et~al.(2022)Ben-Michael, Imai, and Jiang]{ben2022policy}
Eli Ben-Michael, Kosuke Imai, and Zhichao Jiang.
\newblock Policy learning with asymmetric utilities.
\newblock \emph{arXiv preprint arXiv:2206.10479v1}, 2022.

\bibitem[Cui(2021)]{cui2021individualized}
Yifan Cui.
\newblock Individualized decision-making under partial identification: Three
  perspectives, two optimality results, and one paradox.
\newblock \emph{arXiv preprint arXiv:2110.10961}, 2021.

\bibitem[Dawid(2021)]{apd:found}
A.~Philip Dawid.
\newblock Decision-theoretic foundations for statistical causality.
\newblock \emph{Journal of Causal Inference}, 9:\penalty0 39--77, 2021.
\newblock \\{\small\href{http://dx.doi.org/10.1515/jci-2020-0008}{\tt
  DOI:10.1515/jci-2020-0008}}.

\bibitem[Dawid and Musio(2022)]{apd/mm:sef}
A.~Philip Dawid and Monica Musio.
\newblock What can group level data tell us about individual causality?
\newblock In A.~Carriquiry, J.~Tanur, and W.~Eddy, editors, \emph{Statistics in
  the Public Interest: In Memory of Stephen E. Fienberg}, pages 235--256.
  Springer International Publishing, 2022.
\newblock \\ \href{https://doi.org/10.1007/978-3-030-75460-0_13}{\tt DOI:
  10.1007/978-3-030-75460-0\_13}.

\bibitem[Dawid and Senn()]{dawidharm}
Philip~A Dawid and Stephen Senn.
\newblock Personalised decision-making without counterfactuals.
\newblock \emph{Journal of Causal Inference}.

\bibitem[Frangakis and Rubin(2002)]{frangakis2002principal}
Constantine~E Frangakis and Donald~B Rubin.
\newblock Principal stratification in causal inference.
\newblock \emph{Biometrics}, 58\penalty0 (1):\penalty0 21--29, 2002.

\bibitem[Geneletti and Dawid(2011)]{sgg/apd:ett}
Sara~G. Geneletti and A.~Philip Dawid.
\newblock Defining and identifying the effect of treatment on the treated.
\newblock In Phyllis~M. Illari, Federica Russo, and Jon Williamson, editors,
  \emph{Causality in the Sciences}, pages 728--749. Oxford University Press,
  2011.

\bibitem[Greenland and Robins(1986)]{greenland1986identifiability}
Sander Greenland and James~M Robins.
\newblock Identifiability, exchangeability, and epidemiological confounding.
\newblock \emph{International journal of epidemiology}, 15\penalty0
  (3):\penalty0 413--419, 1986.

\bibitem[Kosorok et~al.(2021)Kosorok, Laber, Small, and
  Zeng]{kosorok2021introduction}
Michael~R Kosorok, Eric~B Laber, Dylan~S Small, and Donglin Zeng.
\newblock Introduction to the theory and methods special issue on precision
  medicine and individualized policy discovery, 2021.

\bibitem[Li and Pearl(2019)]{li/pearl:ijcai19}
Ang Li and Judea Pearl.
\newblock Unit selection based on counterfactual logic.
\newblock In \emph{Proceedings of the Twenty-Eighth International Joint
  Conference on Artificial Intelligence, {IJCAI-19}}, pages 1793--1799.
  International Joint Conferences on Artificial Intelligence Organization, 7
  2019.
\newblock {\small\href{https://doi.org/10.24963/ijcai.2019/248}{\tt
  DOI:10.24963/ijcai.2019/248}}.

\bibitem[Li and Pearl(2022{\natexlab{a}})]{li2022probabilities1}
Ang Li and Judea Pearl.
\newblock Probabilities of causation: Role of observational data.
\newblock \emph{arXiv preprint arXiv:2210.08874}, 2022{\natexlab{a}}.

\bibitem[Li and Pearl(2022{\natexlab{b}})]{li2022unit3}
Ang Li and Judea Pearl.
\newblock Unit selection: Case study and comparison with a/b test heuristic.
\newblock \emph{arXiv preprint arXiv:2210.05030}, 2022{\natexlab{b}}.

\bibitem[Mueller and Pearl(2022)]{mueller2022personalized}
Scott Mueller and Judea Pearl.
\newblock Personalized decision making--a conceptual introduction.
\newblock \emph{arXiv preprint arXiv:2208.09558}, 2022.

\bibitem[Murphy(2003)]{murphy2003optimal}
Susan~A Murphy.
\newblock Optimal dynamic treatment regimes.
\newblock \emph{Journal of the Royal Statistical Society: Series B (Statistical
  Methodology)}, 65\penalty0 (2):\penalty0 331--355, 2003.

\bibitem[Orr et~al.(1997)Orr, Pang, Pellegrino, and Siegler]{orr1997use}
Robert~D Orr, Norman Pang, Edmund~D Pellegrino, and Mark Siegler.
\newblock Use of the hippocratic oath: a review of twentieth century practice
  and a content analysis of oaths administered in medical schools in the us and
  canada in 1993.
\newblock \emph{Journal of Clinical Ethics}, 8\penalty0 (4):\penalty0 377--388,
  1997.

\bibitem[Richardson and Robins(2013)]{richardson2013single}
Thomas~S Richardson and James~M Robins.
\newblock Single world intervention graphs (swigs): A unification of the
  counterfactual and graphical approaches to causality.
\newblock \emph{Center for the Statistics and the Social Sciences, University
  of Washington Series. Working Paper}, 128\penalty0 (30):\penalty0 2013, 2013.

\bibitem[Robins(1986)]{robins1986new}
James~M Robins.
\newblock A new approach to causal inference in mortality studies with a
  sustained exposure period—application to control of the healthy worker
  survivor effect.
\newblock \emph{Mathematical modelling}, 7\penalty0 (9-12):\penalty0
  1393--1512, 1986.

\bibitem[Robins(2004)]{robins2004optimal}
James~M Robins.
\newblock Optimal structural nested models for optimal sequential decisions.
\newblock In \emph{Proceedings of the second seattle Symposium in
  Biostatistics}, pages 189--326. Springer, 2004.

\bibitem[Robins et~al.(2007)Robins, Vanderweele, and
  Richardson]{forcinacomment}
James~M. Robins, Tyler~J. Vanderweele, and Thomas~S. Richardson.
\newblock {Comment on ``Causal effects in the presence of non compliance: A
  latent variable interpretation'' by Antonio Forcina}.
\newblock \emph{Metron}, LXIV:\penalty0 288--298, 2007.

\bibitem[Rubin(1974)]{dbr:jep}
Donald~B. Rubin.
\newblock Estimating causal effects of treatments in randomized and
  nonrandomized studies.
\newblock \emph{Journal of Educational Psychology}, 66:\penalty0 688--701,
  1974.

\bibitem[Savage(1951)]{savage1951theory}
Leonard~J Savage.
\newblock The theory of statistical decision.
\newblock \emph{Journal of the American Statistical association}, 46\penalty0
  (253):\penalty0 55--67, 1951.

\bibitem[Stensrud and Sarvet(2022)]{stensrud2022optimal}
Mats~J Stensrud and Aaron~L Sarvet.
\newblock Optimal regimes for algorithm-assisted human decision-making.
\newblock \emph{to come}, 2022.

\bibitem[Stoye(2012)]{stoye2012new}
J{\"o}rg Stoye.
\newblock New perspectives on statistical decisions under ambiguity.
\newblock \emph{Annu. Rev. Econ.}, 4\penalty0 (1):\penalty0 257--282, 2012.

\bibitem[Tian and Pearl(2000)]{tian/pearl:probcaus}
Jin Tian and Judea Pearl.
\newblock Probabilities of causation: Bounds and identification.
\newblock \emph{Annals of Mathematics and Artificial Intelligence},
  28:\penalty0 287--313, 2000.

\bibitem[Tsiatis et~al.(2019)Tsiatis, Davidian, Holloway, and
  Laber]{tsiatis2019dynamic}
Anastasios~A Tsiatis, Marie Davidian, Shannon~T Holloway, and Eric~B Laber.
\newblock \emph{Dynamic treatment regimes: Statistical methods for precision
  medicine}.
\newblock Chapman and Hall/CRC, 2019.

\bibitem[Von~Neumann and Morgenstern(2007)]{von2007theory}
John Von~Neumann and Oskar Morgenstern.
\newblock Theory of games and economic behavior.
\newblock In \emph{Theory of games and economic behavior}. Princeton university
  press, 2007.

\end{thebibliography}

\clearpage
\appendix

\section{Additional Propositions and Proofs}

\added{In this appendix we provide formal justification for results discussed in the main text, including the fusion of experimental and non-experimental data, and the use of current study data for future decision problems.}

\subsection{Preliminaries}

\added{In order to distinguish between experimental and non-experimental data, we introduce an additional variable $R$, a binary indicator of participation within a controlled experiment on $A$ within levels of $L$. We let $O_F \equiv \{R, L, A, Y, Y^{a=1}, Y^{a=0}\}$ denote the full data, and $O \equiv \{R, L, A, Y\}$ denote the observed data.}

\added{Additionally, we will distinguish between probability measures on two dimensions: 1) whether the measure describes a probability distribution for the full counterfactual data or describes a distribution for the observed data; and 2) whether the measure describes a probability distribution in the current study -- used for data analysis--  or in a future setting -- used for decision-making. While the first distinction allows differences between a directly observable distribution and a partially observed causal distribution, the latter allows consideration of possible differences between the present study setting and the future decision setting. }

\added{As such, we let $P^F$ denote an arbitrary joint distribution for the full data, and we let $P\equiv P(P^F)$ denote the subdistribution of $P^F$ for the observed data. Furthermore, we let $P^F_0$ and $P_0$ denote the true such distributions at the time of the study, and further let $P^F_1$ and $P_1$ be the analogous such distributions at the time of a future decision. Note that Definitions \ref{def: relfeat} and \ref{def: Intharm} logically involve expectations evaluated at $P^F_1$, the joint distribution for the full data at the time of a future decision. Likewise, the expectations used in utility maximization, reviewed in Section \ref{sec: decision} (Decision theory), are likewise evaluated at $P^F_1$. Intuitively, $P^F_1$ is the law of the data by which a future decision maker will find themselves governed, and so is the relevant measure for this individual.}

Let $g$ be an arbitrary regime (treatment assignment rule) for $A$ and consider a binary outcome $Y$ taking values in $\{0,1\}$. Suppose the factual data is generated by a set of non parametric structural equations and that all counterfactuals are generated via recursive substitution (for example, as in \citet{richardson2013single}).  

Let $A^{g+}$ represent the assigned value of treatment under regime $g$, where $A$ is that treatment under the regime that generates the factual data. Let $\tau_{P^F}^g$ denote the average treatment effect (ATE) comparing $a=1$ and $g$, at the law $P^F$: $$\mathbb{E}_{P^F}[Y^{a=1}] - \mathbb{E}_{P^F}[Y^g].$$ Let $\tau_{P^F}^0$ denote the special ATE where $g$ is the rule where no one is treated, $a=0$,

$$\mathbb{E}_{P^F}[Y^{a=1}] - \mathbb{E}_{P^F}[Y^{a=0}]$$

and let $\tau_{P^F}$ denote the special case where $g$ is that rule implemented in the factual data 
$$\mathbb{E}_{P^F}[Y^{a=1}] - \mathbb{E}_{P^F}[Y].$$

Likewise, let $\beta_{P^F}^g$ denote the average treatment effect (ATE) comparing $g$ and $a=0$, and let $\beta_{P^F}$ denote the special case where $g$ is the rule implemented in the factual data. 

Let $\mathcal{L}_{P^F}$ be the sharp lower bound for $P^F(S=1)$ when only the parameters $\mathbb{E}_{P^F}[Y^{a=1}]$ and $\mathbb{E}_{P^F}[Y^{a=1}]$ are known.  Let $\mathcal{L}_{P^F}^{*}$ be that sharp lower bound when the parameters $P(Y=y, A=a)$, for all $a$, $y$, are additionally known.

\subsection{Identifying conditional expected potential outcomes under model $\mathcal{M}$}

\added{We define a model $\mathcal{M}$ based on the following assumptions.}

\begin{itemize}
    \item[\textbf{A1.} ] $\mathbb{E}_{P^F_0}[Y^{a} \mid A=a, L=l, R=r] = \mathbb{E}_{P^F_0}[Y \mid A=a, L=l, R=r]$ for all $r, l, a$.
    \item[\textbf{A2.} ] $Y^a \independent_{P_0^F} A \mid L, R=1$ for all $a$.
    \item[\textbf{A3.} ] $(Y^{a=1}, Y^{a=0})  \independent_{P_0^F} R \mid  L$.
    \item[\textbf{A4.} ] $P^F_0 = P^F_1$.    
\end{itemize}
    
\added{Assumption $\textbf{A1}$ states that the expected outcome among patients with treatment $a$, covariate $l$, and trial participation indicator $r$ is equal to that expectation for similar such patients in a trial in which treatment is fixed to $a$, also known as ``distributional consistency''} \citep{apd:found}. \added{Assumption $\textbf{A2}$ states that patients are exchangeable across levels of treatment within strata of $L$ in the trial ($R=1$). Assumption $\textbf{A3}$ states that patients in the trial are exchangeable with patients outside of the trial, within strata of $L$. Assumption $\textbf{A4}$ states that the true law of the full data is stable between the time of data collection ($P^F_0$) and a future decision setting ($P^F_1$).}

\added{We consider $\mathcal{M}$ to be minimal because it comprises only the basic model motifs that are used to: 1) link observed and causal parameters ($\textbf{A1}$); 2) distinguish experimental and non-experimental data ($\textbf{A2}$); 3) justify the fusion of experimental and non-experimental data ($\textbf{A3}$); and 4) justify the use of observed data for future decision settings ($\textbf{A4}$).}
      
We may obtain $\mathbb{E}_{P^F_1}[Y^a \mid L=l]$ from parameters of $P_0$ under $\mathcal{M}$ as follows:

\begin{align}
    \mathbb{E}_{P^F_1}[Y^a \mid L=l] = & \mathbb{E}_{P^F_0}[Y^a \mid L=l]  \label{eq: proof1} \\
                                     = & \mathbb{E}_{P^F_0}[Y^a \mid L=l, R=1] \nonumber \\
                                     = & \mathbb{E}_{P^F_0}[Y^a \mid L=l, A=a, R=1]  \nonumber\\
                                     = & \mathbb{E}_{P_0}[Y \mid L=l, A=a, R=1],  \nonumber
\end{align}
assuming all terms are well defined. The first equality follows from equivalence of the full law at the time of data collection $P_0^F$ and the full law at the future decision time point $P_1^F$ (\textbf{A4.}), the second equality follows from the exchangeability of trial and non-trial participants given $L$ (\textbf{A3.}), the third follows from the exchangeability of treated and untreated participants in the trial given $L$ (\textbf{A2.}), and the fourth follows from consistency (\textbf{A1.}).

\subsection{On the Interventionist use of treatment intention as a relevant covariate}

\added{Here we define an elaborated model $\mathcal{M}^*$ whose assumptions permit the fusion of experimental and non-experimental data for improved decision making under an Interventionist approach. We let $\mathcal{M}^*$ denote the model defined by $\mathcal{M}$ supplemented with assumptions \textbf{B1.} to \textbf{B3}. We let $P^{F*}$ denote an arbitrary law of the full data including the intention to treat variable $A^*$.}

\begin{itemize}
    \item[\textbf{B1.} ] $P^*_1(A=A^* \mid R=0)=1$.
    \item[\textbf{B2.} ] $(Y^{a=1}, Y^{a=0}, A^*)  \independent_{P_0^F} R \mid  L$.
    \item[\textbf{B3.} ] $P^{F*}_0 = P^{F*}_1$.
\end{itemize}
    
\added{Assumption \textbf{B1.} asserts that a patient's intention to treat variable ($A^*$) agrees almost surely with the treatment actually received ($A$)  in the non-experimental setting ($R=0$), and thus the intention to treat variable is indirectly measured. Assumption \textbf{B2.} differs from \textbf{A3.} by asserting that trial participation is additionally independent of a patient's intention to take treatment $A^*$, conditional on covariates. Assumption \textbf{B3.} differs from \textbf{A4.} in asserting that the elaborated law of the full data (that includes $A^*$) is also stable between the time of data collection ($P^{F*}_0$) and a future decision setting ($P^{F*}_1$).}

\added{The following proposition asserts that these parameters are indeed identified under $\mathcal{M}^*$. }

\begin{proposition}

\label{prop:s2}
Under $\mathcal{M}^*$,  $P^{F*}_1(Y^{a}=y \mid A^*=a^*, L=l)$  is identified for all $(a,a^* = 0,1)$.  Specifically, 

\begin{align*}
  P^{F*}_1(Y^{a}=y \mid A^*=a^*, L=l) =   
\begin{cases}
    P(Y=y \mid A=a,  L=l, R=1) & \text{for } a=a^* \\
    \frac{P(Y=y \mid  A=a,  L=l, R=1) - P(Y=y, A=a \mid L=l, R=0)}{P(A=a^* \mid L=l, R=0)}  & \text{for } a\neq a^*. 
\end{cases}
\end{align*}

\end{proposition}

\added{Related propositions have appeared in} \citet{stensrud2022optimal, forcinacomment,sgg/apd:ett,apd/mm:sef}.

\begin{proof}
  When $a=a^*$, arguments proceed identically as in \eqref{eq: proof1}. When $a \neq a^*$, we have by the law of total probability that $P^{F*}_1(Y^{a}=y \mid A^*=a^*, L=l)$ equals

  $$\frac{P^{F*}_1(Y^{a}=y \mid  L=l) - P^{F*}_1(Y^{a}=y, A^*=a \mid L=l)}{P^{F*}_1(A^*=a^* \mid L=l)}.$$

Then, the left hand term in the numerator is identified by the experimental data under the assumptions in $\mathcal{M}$ alone, and the remaining terms are identified by the non-experimental data using assumptions unique to $\mathcal{M}^*$. For example, we have :

\begin{align*}
P^{F*}_1(Y^{a}=y, A^*=a \mid L=l) = &  P^{F*}_0(Y^{a}=y, A^*=a \mid L=l), \\
= &  P^{F*}_0(Y^{a}=y, A^*=a \mid L=l, R=0), \\
= &  P^{F*}_0(Y^{a}=y, A=a \mid L=l, R=0), \\
= &  P_0(Y=y, A=a \mid L=l, R=0). 
\end{align*}

In the above, the first equality follows from \textbf{B3.}, the second by \textbf{B2.}, the third by \textbf{B1.}, and the fourth by laws of probability and \textbf{A3.}
\end{proof}
The above proof relies on $A$ (and so $A^*$) being binary, but $Y$ need not be.

\subsection{Propositions \ref{prop:s3} and \ref{prop:s4}}

\begin{proposition}\label{prop:s3}
    
The average treatment effect comparing $a=1$ and $g$ is a valid lower bound for the probability of harm. That is, for all laws $P^F$:

\begin{align*}
  P^{F}(S=1) \geq  \tau_{P^F}^g.  
\end{align*}

\end{proposition}

\begin{proof}

    \begin{align*}
        P^F(Y^g=1) = & P^F(Y^g=1, A^{g+}=1) + P(Y^g=1, A^{g+}=0) \\
                 = & P^F(Y^{a=1}=1, A^{g+}=1) + P^F(Y^{a=0}, A^{g+}=0) \\
                 = & \Big\{P^F(S=1, A^{g+}=1) +  P^F(S=2, A^{g+}=0)\Big\} + P^F(S=3). 
    \end{align*}

The first equality follows by laws of probability, the second by recursive substitution (consistency), and and third follows by laws of probability. Let $\epsilon$ denote the term $\{P^F(S=1, A^{g+}=1) +  P^F(S=2, A^{g+}=0)\}$. Then 

    \begin{align*}
        \tau_{P^F}^g = & P^F(S=1)  - \{P^F(S=1, A^{g+}=1) +  P^F(S=2, A^{g+}=0)\}\\
               = & P^F(S=1)  - \epsilon.
    \end{align*}

The first equality follows by definition (and cancellation of $P^F(S=3)$), and the final by definition. Because $\epsilon$ must be positive for all laws $P^F$, then $\tau_{P^F}^g$ is a valid lower bound for $P^F(S=1)$. 
\end{proof}

\begin{proposition}\label{prop:s4}
    
When strong exchangeability holds under $P^F$ relative to regime $g$, $$\{Y^{a=1}, Y^{a=0}\} \independent_{P^F} A^{g+}$$ then $\tau_{P^F}^0 \geq \tau_{P^F}^g$. 

\end{proposition}

\begin{proof}

    \begin{align*}
    \tau_{P^F}^g =& P^F(S=1) - \{P^F(S=1, A^{g+}=1) +  P^F(S=2, A^{g+}=0)\} \\
           =& P^F(S=1) - \{P^F(S=1)P^F(A^{g+}=1) +  P^F(S=2)P^F(A^{g+}=0)\} \\
           =& P^F(A^{g+}=0)\{P^F(S=1) -  P^F(S=2)\} \\
           =& \tau_{P^F}^0P^F(A^{g+}=0).
    \end{align*}

The first equality follows by Proposition \ref{prop:s3}, the second by strong exchangeability, and the final two by laws of probability. Because $P^F(A^{g+}=0)$ is bounded to the unit interval for all laws $P^F$, then the claim of Proposition \ref{prop:s4} follows immediately. 
\end{proof}

Arguments are identical when all terms additionally condition on the event $L=l$.

\subsection{Proposition \ref{prop:s5}}

\citet{tian/pearl:probcaus} showed that the lower bound for $P^F(S=1)$ is improved by the addition of non-experimental data, that is, $$\mathcal{L}^*_{P^F} > \mathcal{L}_{P^F}$$ whenever:

\begin{align}
    \text{max}
        \begin{Bmatrix}
        & \tau_{P^F}, \\
        & \beta_{P^F} 
        \end{Bmatrix} >
    \text{max}\begin{Bmatrix}
        & 0, \\
        & \tau_{P^F}^0
        \end{Bmatrix}. \label{eq: Lowerimprov}
\end{align}

%

Let $\tau_{P^F, A=1}$ denote the average treatment effect among the treated, 

$$\tau_{P^F, A=1}\coloneqq \mathbb{E}_{P^F}(Y^{a=1} \mid A=1) - \mathbb{E}_{P^F}(Y^{a=0} \mid A=1),$$ and let $\tau_{P^F, A=0}$ denote the average treatment effect among the untreated, $A=0$.

\begin{proposition}\label{prop:s5}

Condition \eqref{eq: Lowerimprov} holds if and only if the following condition holds:
\begin{align}
    \text{sign} \Big\{\tau_{P^F, A=1} \Big\} \not\equiv 
    \text{sign} \Big\{\tau_{P^F, A=0} \Big\}. \label{eq: supercate_ineq} 
\end{align}

\end{proposition}

\begin{proof}
Note that \eqref{eq: supercate_ineq} holds if either $\tau_{P^F} > \{ 0, \tau_{P^F}^0\}$ or $\beta_{P^F} > \{0, \tau_{P^F}^0\}$. Consider the following implications:
 
\begin{align*}
    & \tau_{P^F} > \tau_{P^F}^0 \\
    \iff & P^F(S=1)  - \{P^F(S=1, A=1) +  P^F(S=2, A=0)\} > P^F(S=1)  - P^F(S=2) \\ 
    \iff &  P^F(S=2) > P^F(S=1, A=1) +  P^F(S=2, A=0) \\ 
    \iff &  P^F(S=2, A=1) > P^F(S=1, A=1)  \\ 
    \iff &  P^F(S=2 \mid A=1) > P^F(S=1 \mid A=1) \\
    \iff &  \tau_{P^F, A=1} < 0.
\end{align*}
    
Suppose also that $\tau_{P^F}>0$. 

\begin{align*}
    & \tau_{P^F}>0 \\
    \iff & P^F(S=1)  > P^F(S=1, A=1) +  P^F(S=2, A=0) \\ 
    \iff & P^F(S=1, A=0)  >  P^F(S=2, A=0) \\ 
    \iff & P^F(S=1 \mid A=0)  >  P^F(S=2 \mid A=0) \\
    \iff &  \tau_{P^F, A=0} > 0.
\end{align*}

Thus, $\tau_{P^F} > \{ 0, \tau_{P^F}^0\}$ if and only if both $\tau_{P^F, A=1} < 0$ and $\tau_{P^F, A=0} > 0$. By isomorphic arguments, we can show that $\beta_{P^F} > \{ 0, \tau_{P^F}^0\}$ if and only if both $\tau_{P^F, A=1} > 0$ and $\tau_{P^F, A=0} < 0$. Then the claim of proposition \ref{prop:s5} follows immediately.
    
\end{proof}

Arguments are identical when all terms additionally condition on the event $L=l$. Thus, it follows that, when $A=A^*$ with probability 1 (as under $\mathcal{M}^*$) and  \eqref{eq: Lowerimprov} holds, the expected outcome for the optimal regime that is a function of $L$ will be inferior to that for the optimal regime that is additionally a function of $A^*$. Thus, the argument in the main text follows: under $\mathcal{M}^*$,  non-experimental data will improve the lower bound on the probability of \textit{counterfactual} harm if and only if an \textit{interventionist} would obtain a greater expected utility through the use of such data.

\section{Review of principal strata bounds with experimental and non-experimental data}

\citet{tian/pearl:probcaus} note that, for an arbitrary conditioning event $X=x$,  constraints of the type $$P^F_1(S=1, A=1 \mid X=x) +  P^F_1(S=3, A=1 \mid X=x) = P^F_1(Y=1, A=1 \mid X=x)$$ will impose constraints on the principal strata probabilities, for example, that  $P^F_1(S=1 \mid X=x)$ must be greater than $P_1(Y=1 \mid X=x) - P_1^F(Y^{a=0}=1 \mid X=x)$. Leveraging experimental data under model $\mathcal{M}$, we likewise have that $P^F_1(S=1 \mid L=l) = P^F_1(S=1 \mid L=l, R=1)$ must be greater than $P_1(Y=1 \mid L=l, R=1) - P_1^F(Y^{a=0}=1 \mid L=l)$. We also have 
$P^F_1(S=1 \mid L=l) = P^F_1(S=1 \mid L=l, R=0)$ must be greater than $P_1(Y=1 \mid L=l, R=0) - P_1^F(Y^{a=0}=1 \mid L=l)$. The results of \citet{tian/pearl:probcaus} motivate the following lower bound on $P^F_1(S=s \mid L=l)$:

\begin{align}
    P^F_1(S=1 \mid L=l) \geq 
        \text{max} \begin{Bmatrix}
            & 0 \\
          &  P_1^F(Y^{a=1}=1 \mid L=l) - P_1^F(Y^{a=0}=1 \mid L=l),  \\
          &  P_1(Y=1 \mid L=l, R=1) - P_1^F(Y^{a=0}=1 \mid L=l),  \\  
          &  P_1^F(Y^{a=1}=1 \mid L=l) - P_1(Y=1 \mid L=l, R=1),  \\
          &  P_1(Y=1 \mid L=l, R=0) - P_1^F(Y^{a=0}=1 \mid L=l),  \\  
          &  P_1^F(Y^{a=1}=1 \mid L=l) - P_1(Y=1 \mid L=l, R=0).  \\
         \end{Bmatrix}
\end{align}

Using identification results from the experimental data, we express these constraints fully in terms of observed data:

\begin{align}
    P^F_1(S=1 \mid L=l) \geq 
        \text{max} \begin{Bmatrix}
          & 0, \\
          &  P_0(Y=1 \mid A=1, L=l, R=1) - P_0(Y=1 \mid  A=0, L=l, R=1),  \\
          &  P_0(Y=1 \mid  L=l, R=1) - P_0(Y=1 \mid  A=0, L=l, R=1),  \\  
          &  P_0(Y=1 \mid A=1, L=l, R=1) - P_0(Y=1 \mid  L=l, R=1), \\
          &  P_0(Y=1 \mid  L=l, R=0) - P_0(Y=1 \mid  A=0, L=l, R=1),  \\  
          &  P_0(Y=1 \mid A=1, L=l, R=1) - P_0(Y=1 \mid  L=l, R=0). 
         \end{Bmatrix}
\end{align}

However, because of the constraint that $P_0(Y=1 \mid  L=l, R=1)$ is a convex combination of $P_0(Y=1 \mid  A=0, L=l, R=1)$ and $P_0(Y=1 \mid A=1, L=l, R=1)$, then the third and fourth terms in the set are guaranteed to be less or equal to than the second term, and so they do not effectively sharpen the bounds. However, when   $Y^a \not\independent_{P_0^F} A \mid L, R=0$ for some $a$, as would be expected in non-experimental data, then no such constraint on $P_0(Y=1 \mid  L=l, R=0)$ is imposed and these additional constraints (the fifth and sixth terms) may indeed sharpen the bounds. Of course, such constraints may only be leveraged when investigators additionally take advantage of (confounded) non-experimental data. We also emphasize that these bounds will only hold under models like $\mathcal{M}$, in which the CATEs are preserved between both experimental and non-experimental data, but that the conditional joint distribution of treatment $A$ and potential outcomes differ.

\section{A note on potential outcomes notation}

Although we use a notation that often corresponds to assumptions about the existence of counterfactual variables,  we strongly emphasize that such assertions are unnecessary for the \textit{interventionist} approach. We use this notation in the first place because the \textit{counterfactual} approach depends inextricably on the existence of counterfactuals. We use them throughout to reduce notational burden, but we encourage readers to interpret any parameter defined by a \textit{single} intervention as simply the analogous parameter that would be observed in the arm of a hypothetical controlled experiment corresponding to that intervention. For example, the expected potential outcome under treatment $a$, $Y^a$, among patients with covariates $L=l$ should be interpreted as the expected outcome among patients with $L=l$ in a hypothetical controlled experiment where $A$ is fixed to $a$. As a second example, claims of (conditional) independencies between single-intervention counterfactuals and treatments, for example $Y^{a=1} \independent A$, should simply be interpreted as claims that patients with different treatments in the observed data are (conditionally) exchangeable with respect to the outcome $Y$. 
\end{document}